\documentclass{article}
\usepackage{arxiv}

\usepackage[svgnames]{xcolor}
\usepackage{graphicx}
\usepackage[font=small, skip=5pt, belowskip=0pt]{caption}
\usepackage[font=footnotesize]{subcaption}
\usepackage{multirow}
\usepackage{array}
\usepackage{booktabs}
\usepackage{pdflscape}
\usepackage[ruled,vlined,linesnumbered]{algorithm2e}
\usepackage[T1]{fontenc}
\usepackage[utf8]{inputenc}
\usepackage{cite}
\usepackage{url}
\usepackage{comment}
\usepackage[normalem]{ulem}
\usepackage{mleftright,xparse}
\usepackage{amssymb,amsmath,amsfonts,amsthm,bm,mathtools,cuted,bbold}
\usepackage{siunitx}
\usepackage{tikz}
\usetikzlibrary{math, shapes}
\usepackage{pgfplots}
\usepgfplotslibrary{groupplots}
\pgfplotsset{
compat=1.13,
legend style={font=\footnotesize, fill opacity=0.7,  draw opacity=1, text opacity=1, draw=white!15!black, legend cell align=left, align=left},
width=6cm, 
height=6cm,
yminorticks=false,
xminorticks=false,
label style={font=\small},
title style={font=\small},
tick style={color=black},
tick align=outside,
tick pos=left,
tick label style={font=\footnotesize},
grid style={line width=.1pt, draw=gray!20},
major grid style={line width=.1pt,draw=gray!20},
}

\usepackage{hyperref}
\hypersetup{
  colorlinks   = true, 
  urlcolor     = blue, 
  linkcolor    = black, 
  citecolor   = black 
}

\usepackage[acronym]{glossaries}
\newacronym{3d}{3D}{three dimensional}
\newacronym{aoa}{AoA}{angle of arrival}
\newacronym{ack}{ACK}{acknowledgment}
\newacronym{aod}{AoD}{angle of departure}
\newacronym{ap}{AP}{access point}
\newacronym{awgn}{AWGN}{additive white gaussian noise}
\newacronym{b5g}{B5G}{Beyond-5G}
\newacronym{bsw}{CB-BSW}{codebook-based beam sweeping}
\newacronym[plural=BSs, firstplural=base stations (BSs)]{bs}{BS}{base station}
\newacronym{ctrl}{ctl}{control}
\newacronym{cc}{CC}{control channel}
\newacronym{cdf}{CDF}{cumulative density function}
\newacronym{ce}{CE}{channel estimation}
\newacronym{cmos}{CMOS}{complementary metal-oxide semiconductor}
\newacronym{clt}{CLT}{central limit theorem}
\newacronym{cpu}{CPU}{central processing unit}
\newacronym{oce}{OPT-CE}{optimization based on channel estimation}
\newacronym{obcc}{OB-CC}{out-of--band \gls{cc}}
\newacronym{csi}{CSI}{channel state information}
\newacronym{dc}{DC}{data channel}
\newacronym{dl}{DL}{downlink}
\newacronym{dft}{DFT}{discrete Fourier transform}
\newacronym{doa}{DoA}{direction-of-arrival}
\newacronym{emf}{EMF}{electromagnetic field}
\newacronym{em}{EM}{electromagnetic}
\newacronym{es}{ES}{edge server}
\newacronym{fdd}{FDD}{frequency-division duplexing}
\newacronym{fdm}{FDM}{frequency division multiplexing}
\newacronym{fp}{FP}{fractional program}
\newacronym[plural=HRISs, firstplural=Hybrid Reconfigurable Intelligent Surfaces (HRISs)]{hris}{HRIS}{hybrid reconfigurable intelligent surface}
\newacronym{ibcc}{IB-CC}{in-band \gls{cc}}
\newacronym{ibno}{IB-no}{IB-no}
\newacronym{ibwf}{IB-wf}{IB-wf}
\newacronym{ios}{IoS}{Internet-of-Surfaces}
\newacronym{iot}{IoT}{Internet-of-Things}
\newacronym{iid}{i.i.d.}{independently identically distributed}
\newacronym[plural=KPIs, firstplural=key performance indicators (KPIs)]{kpi}{KPI}{key performance indicator}
\newacronym{lf}{LF}{low frequency}
\newacronym{ls}{LS}{least squares}
\newacronym{los}{LoS}{line-of-sight}
\newacronym{mcs}{MCS}{modulation and coding scheme}
\newacronym{mec}{MEC}{Mobile Edge Computing}
\newacronym{mimo}{MIMO}{multiple-input multiple-output}
\newacronym{miso}{MISO}{multiple-input single-output}
\newacronym{ml}{ML}{machine learning}
\newacronym{mmse}{MMSE}{minimum mean squared error}
\newacronym{mrt}{MRT}{maximum-ratio transmission}
\newacronym{mse}{MSE}{mean squared error}
\newacronym{nlos}{NLoS}{non-line-of-sight}

\newacronym{qos}{QoS}{quality of service}
\newacronym{ofdm}{OFDM}{orthogonal frequency-division multiplexing}

\newacronym{pdf}{pdf}{probability distribution function}
\newacronym{pla}{PLA}{planar linear array}
\newacronym{pap}{P\&P}{plug-and-play}
\newacronym{ppp}{PPP}{Poisson point process}
\newacronym{ra}{RA}{resource allocation}
\newacronym[plural=RISs, firstplural=Reconfigurable Intelligent Surfaces (RISs)]{ris}{RIS}{reconfigurable intelligent surface}
\newacronym{risc}{RISC}{\gls{ris} controller}
\newacronym{rf}{RF}{radio frequency}
\newacronym{rmse}{RMSE}{root-mean-square error}
\newacronym{rss}{RSS}{received signal strength}
\newacronym{se}{SE}{spectral efficiency}
\newacronym{sdp}{SDP}{semidefinite programming}
\newacronym{sdr}{SDR}{semidefinite relaxation}
\newacronym{simo}{SIMO}{single-input multiple-output}
\newacronym{sinr}{SINR}{signal-to-interference-plus-noise ratio}
\newacronym{smse}{SMSE}{sum mean squared error}
\newacronym{snr}{SNR}{signal-to-noise ratio}
\newacronym{soa}{SoA}{state-of-the-art}
\newacronym{sre}{SRE}{smart radio environment}
\newacronym{toa}{ToA}{time-of-arrival}
\newacronym{tdd}{TDD}{time division multiplex}
\newacronym{tti}{TTI}{transmission time interval}
\newacronym[plural=UEs, firstplural=user equipments (UEs)]{ue}{UE}{user equipment}
\newacronym{ul}{UL}{uplink}
\newacronym{ula}{ULA}{uniform linear array}
\newacronym{urllc}{URLLC}{ultra-reliable low-latency communications}

\newcommand{\ie}{\emph{i.e.}}     
\newcommand{\eg}{\emph{e.g.}}     
\newcommand{\T}{^{\mathsf{T}}}     
\renewcommand{\H}{^{\mathsf{H}}}   
\newcommand{\mc}[1]{\mathcal{#1}}   
\newcommand{\mb}[1]{\mathbf{#1}}    

\newcommand{\alg}{_\mathrm{alg}} 
\newcommand{\ini}{_\mathrm{ini}} 

\newcommand{\ce}{_\mathrm{ce}}
\newcommand{\ra}{_\mathrm{ra}}
\newcommand{\set}{_\mathrm{set}}
\newcommand{\pay}{_\mathrm{pay}}

\newcommand{\ctl}{\mathrm{ctl}}

\definecolor{gold}{rgb}{1.0, 0.84, 0.0}
\definecolor{green2}{rgb}{0.0, 1.0, 0.0}
\definecolor{purple2}{rgb}{0.63, 0.36, 0.94}
\definecolor{amaranth}{rgb}{0.9, 0.17, 0.31}
\definecolor{cadmiumgreen}{rgb}{0.0, 0.42, 0.24}
\definecolor{gold}{rgb}{0.85,.66,0}

\newtheorem{proposition}{Proposition}

\begin{document}
\title{Control Aspects for Using RIS in Latency-Constrained Mobile Edge Computing}

\author{        
    \parbox{\textwidth}{\centering
      Fabio Saggese$^\dagger$,
      Victor Croisfelt$^\dagger$,
      Francesca Costanzo$^{\circ}$ ,
      Junya Shiraishi$^\dagger$,
      Rados\l{}aw Kotaba$^\dagger$, \\
      Paolo Di Lorenzo$^\circ$, 
      and Petar Popovski$^\dagger$       
  }%
  \\
  \parbox{\textwidth}{
    \centering
    $^\dagger$Department of Electronic System, Aalborg University, Denmark \\
    $^\circ$Department of Information Engineering, University of Rome ``La Sapienza'', Italy \\  
    $^\circ$Consorzio Nazionale Interuniversitario per le Telecomunicazioni (CNIT), Parma, Italy \\  
    email: $^\dagger$\{fasa, vcr, jush, rak, petarp\}@es.aau.dk,
    $^\circ$\{francesca.costanzo, paolo.dilorenzo\}@uniroma1.it   
  }%
    \thanks{This work was partly supported by the Villum Investigator grant ``WATER'' from the Villum Foundation, Denmark, and by the EU H2020 RISE-6G project under grant number 101017011. (Corresponding author: Fabio Saggese, email: fasa@es.aau.dk)
    }
}

\maketitle

\begin{abstract}    
    This paper investigates the role and the impact of control operations for dynamic mobile edge computing (MEC) empowered by Reconfigurable Intelligent Surfaces (RISs), in which multiple devices offload their computation tasks to an access point (AP) equipped with an edge server (ES), with the help of the RIS. While usually ignored, the control aspects related to channel estimation (CE), resource allocation (RA), and control signaling play a fundamental role in the user-perceived delay and energy consumption. In general, the higher the resources involved in the control operations, the higher their reliability;
    however, this introduces an overhead, which reduces the number of resources available for computation offloading, possibly increasing the overall latency experienced. Conversely, a lower control overhead translates to more resources available for computation offloading but impacts the CE accuracy and RA flexibility. This paper establishes a basic framework for integrating the impact of control operations in the performance evaluation of the RIS-aided MEC paradigm, clarifying their trade-offs through theoretical analysis and numerical simulations.
\end{abstract}


\section{Introduction}
\label{sec:intro}
Beyond 5G networks are fundamental enablers of new services and applications (including verticals), such as Industry 4.0, \gls{iot}, and autonomous driving \cite{ahmadi20195g,6Gstrinati}, which typically require massive data processing, high reliability, and low end-to-end delays. In this context, a key technical enabler is \gls{mec}, whose aim is to allow proximity access to (albeit limited) cloud functionalities (\eg, computing and storage resources) at the edge of the wireless network~\cite{barbarossa2014communicating}. In \gls{mec}, \glspl{ue} offload heavy computational tasks to nearby processing units or \glspl{es}, typically placed close to \glspl{ap} of the radio access network, to improve energy consumption, reduce latency, and/or running sophisticated applications that would be impossible to execute at the \gls{ue} side. However, 
the available network resources must be jointly optimized and orchestrated to provide the end users with a satisfactory \gls{qos}. This problem has sparked a lot of interest in recent literature, with several works that have proposed joint resource optimization frameworks in both static and dynamic \gls{mec} scenarios, \eg,~\cite{barbarossa2014communicating,Chen2019,Merluzzi2020URLLC,HanChen2020,merluzzi2021discontinuous}.  

In \gls{mec} systems, good wireless connectivity is a necessary condition for ensuring the required \gls{qos}. In the presence of poor wireless channel conditions, the performance of \gls{mec} systems might be severely hindered because of lower offloading rates and consequent under-exploitation of the computation capabilities of the \gls{es}. 
A substantial performance boost can be achieved by exploiting \glspl{ris}, an emerging technology that has gained significant attention in recent years due to its ability to shape and control the wireless communication environment~\cite{bjornson2022reconfigurable}. To this aim, several works have already investigated the optimization of \gls{ris}-aided \gls{mec} systems, considering both static and dynamic computation offloading~\cite{bai2021reconfigurableWC, chu2020intelligent, huang2021reconfigurable, hu2021reconfigurable, DiLorenzo2022}. These methods aim to jointly optimize the \gls{ris} configuration, the \gls{ap} and \glspl{ue} communication parameters, and the \gls{es} \gls{cpu} resources to provide a target \gls{qos}. 
In~\cite{chu2020intelligent}, the authors focus on maximizing the number of processed bits for computation offloading. 
Furthermore,~\cite{huang2021reconfigurable} explores the utilization of \glspl{ris} to enhance the performance of machine learning tasks at the \gls{es}. 
\cite{hu2021reconfigurable} suggests optimization-based and data-driven solutions for \gls{ris}-aided multi-user \gls{mec}, aiming to maximize the total completed task-input bits of all \glspl{ue} within limited energy budgets. 
Lastly,~\cite{DiLorenzo2022} delves into the dynamic optimization of a \gls{ris}-empowered dynamic \gls{mec} scenario in which the \glspl{ue} continuously generate data for offloading while wireless channel conditions evolve through time. The optimization aims to ensure minimum energy consumption under an average latency constraint.

To the best of our knowledge, the control operations needed for using \glspl{ris} in \gls{mec} systems and their impact on the \gls{qos} experienced by the end user has never been investigated within the existing literature. The analysis of those procedures is central to enabling \gls{ris}-aided \gls{mec} services in future wireless networks. Among the first works on the subject is~\cite{saggese2023impact}, which proposes a general method to include the control aspects in the performance evaluation of \gls{ris}-aided networks. By using similar considerations, this work defines and quantifies the impact of control operations on the performance of \gls{ris}-empowered \gls{mec} systems by focusing on the trade-offs and interactions between the amount of resources invested in the control procedures \emph{vs} data transmission and computation offloading at the \gls{mec}. The analysis uses the \gls{mec} offloading strategy proposed by~\cite{DiLorenzo2022}. Our results show that control errors may disrupt the \gls{mec} service if not properly considered during the system design. While the results underline the necessity for high reliability of the control signaling, it is seen that the \gls{ce} procedure strongly degrades the performance. On one side, its high overhead reduces the end users' \gls{qos}, limiting the time reserved for offloading operations; on the other, errors in \gls{csi} estimation strongly affect the \gls{ra}, and thus the optimality of the allocated resources.

\paragraph*{Notation}
Boldface lower, $\mb{x}$, and capital, $\mb{X}$, letters  denote vector and matrices, respectively; $\mb{0}_n$ is an $n$-size zero vector; $\mb{I}_n$ is the $n\times n$ identity matrix; $\lVert\cdot\rVert$ is the Euclidean norm; $\circ$ is the element-wise multiplication. Calligraphic letters, $\mc{A}$, denote sets; $|\cdot|$ denotes the absolute value if applied to a scalar or the cardinality if applied to a set. 
$\mc{CN}(\bm{\mu}, \bm{\Sigma})$ is the complex Gaussian distribution with mean value $\bm{\mu}$ and covariance matrix $\bm{\Sigma}$; $\mathbb{E}_{\sim x}[\cdot]$ is the expected value w.r.t. $x$.

\section{System Model}
\label{sec:model}
\begin{figure}
    \centering
    \includegraphics[width=\columnwidth]{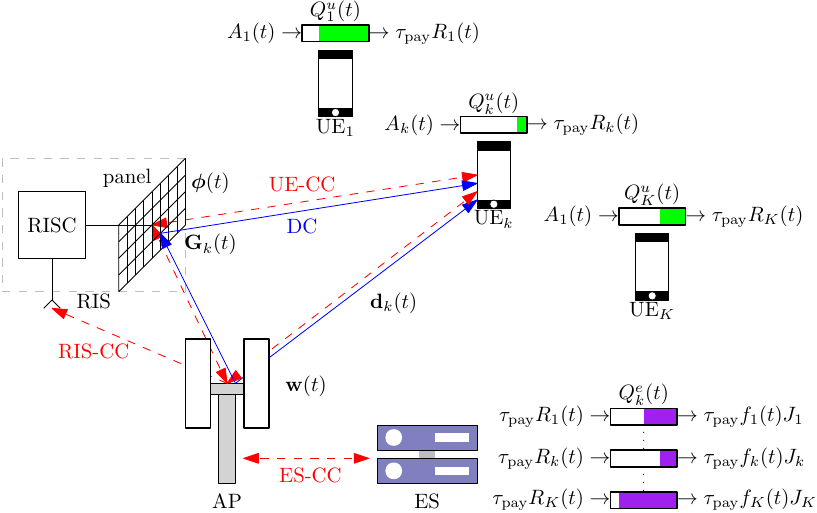}
    \caption{
    \Gls{ris}-aided \gls{mec} system of interest: $K$ single-antenna \glspl{ue} offload their data to an \gls{es}, connected to a multi-antenna \gls{ap} aided by a \gls{ris}, by transmitting a portion of their local queues $Q_k^u(t)$ in each slot $t$. The \gls{es} stores the data in remote queues $Q^e_k(t)$ and performs computation over them using CPU frequencies $f_k(t)$. To account for the control operations that dictate how the data is offloaded, a \textcolor{blue}{data channel (DC)} and three \textcolor{red}{control channels (CCs)} are considered.
    }
    \label{fig:scenario}
    \vspace{-0.5cm}
\end{figure}

We consider the \gls{ris}-aided \gls{mec} system depicted in Figure~\ref{fig:scenario}, where $K$ single-antenna \glspl{ue} want to offload computational-heavy tasks to an \gls{es} endowed with an $M$-antenna \gls{ap} and aided by an \gls{ris} equipped with $N$-elements. 
For simplicity of notation, we collect the sets of \glspl{ue} and \gls{ap} antennas in $\mc{K}=\{1, \dots, K\}$ and $\mc{M} = \{1, \dots, M\}$, respectively. We assume that the \glspl{ue} are multiplexed according to a \gls{fdm} technique, enabling them to transmit or receive simultaneously without interfering with each other~\cite{DiLorenzo2022}. Then, we let $B$ be the total bandwidth allocated for communication between the \gls{ap} and the \glspl{ue}. We denote as $B_k$ the bandwidth allocated to the $k$-th \gls{ue}, such that $B = \sum_{k\in\mc{K}} B_k$.

Similar to~\cite{DiLorenzo2022}, we consider that the offloading is dynamic; that is, the \glspl{ue} continuously generate data that need to be processed by the \gls{es}. Hence, the system ends when the offloaded data of all \glspl{ue} are entirely processed by the \gls{es}. In such a dynamic system, time is organized into \emph{slots} of equal duration $\tau\in\mathbb{R}_{+}$, indexed by $t\in\{1,2,\dots\}$. A slot is then further divided into control and payload parts as
\begin{equation}
    \tau = \tau_\ctl + \tau\pay.
    \label{eq:slot}
\end{equation}
Within the control time or \textit{control overhead} $\tau_\ctl$, \textit{control operations} are performed, which include: a) the \gls{ce} used to obtain the \gls{csi} at the \gls{ap}, b) the \gls{ra} carried out by the \gls{es} and used to optimize the transmission parameters for the \glspl{ue}, the processing parameters at the \gls{es}, and the \gls{ris} properties, and c) the control signaling used to exchange relevant information among all the nodes involved. The payload time $\tau\pay$ comprises the actual transmission/offloading of data by the \glspl{ue}. 

We aim to \emph{minimize the total energy consumption of the system while keeping the offloading latency perceived by the \glspl{ue} under a threshold}. In Section~\ref{sec:protocol}, we give detailed information on the protocol that dictates how the offloading service occurs on a slot basis when achieving this goal and considering the control operations. Below, we introduce the notation and models required to describe the protocol.

\subsection{AP and RIS modeling}
\label{sec:ris-model}
The \gls{ap} enforces the synchronization of the system and controls the behavior of the \gls{ris}~\cite{DiLorenzo2022}. The \gls{ap} is assumed to have analog beamforming capabilities, \ie, it can load a beamforming vector $\mb{w}(t)\in\mathbb{C}^M$ with $\lVert \mb{w}(t) \rVert = 1$ to focus its transmission at the $t$-th slot. The set of possible available beamforming vectors forms the \gls{ap} codebook denoted as $\mc{C}_\mathrm{ap}$, modeled from accurate antenna patterns as in~\cite{APcodebook}, such that $\mb{w}(t)\in\mc{C}_\mathrm{ap}$. 

The \gls{ris} device is made of a \gls{risc} and a \gls{ris} panel~\cite{saggese2023impact}. The former receives control information from the \gls{ap} and commands the changes of the electromagnetic properties of the latter. If activated, each element of the \gls{ris} panel can be tuned to impose a \emph{phase-shift} variation on the impinging wave. At every slot $t$, the set of the states of all the elements forms the so-called \emph{\gls{ris} configuration}, denoted as
\begin{equation} \label{eq:phi}
    \bm{\phi}(t) = [\alpha_1(t) e^{j\phi_1(t)}, \dots,  \alpha_N(t) e^{j\phi_N(t)}]\T \in\mathbb{C}^N,    
\end{equation}
where $\alpha_n(t) \in \{0,1\}$ and $\phi_n(t)\in\left\{\frac{2 i \pi}{2^b}\right\}_{i=0}^{2^b-1}$ represent the activation state and the ($b$-bit quantized) phase shift of the $n$-th \gls{ris} element, respectively. Without loss of generality, we assume the \gls{ris} needs $\tau_s$ seconds to load a new configuration due to its hardware characteristics. Moreover, we assume the \gls{risc} has a single antenna with a single receiver \gls{rf} chain to process the incoming control information from the \gls{ap}. In particular, the \gls{ris} has a \emph{wide-width beam \gls{ctrl} configuration} $\bm{\phi}_\ctl\in\mathbb{C}^{N}$ that can boost the overall coverage of the \gls{ap} in a certain area of interest where the \glspl{ue} might be located. We assume the \gls{ris} always loads the control configuration when control signaling is exchanged between \gls{ap} and \glspl{ue}. The \gls{ris} is also equipped with a \emph{\gls{ce} codebook} $\mc{C}\ce$ that contains $|\mc{C}\ce|=C\ce$ pre-defined configurations $\boldsymbol{\phi}\in\mc{C}\ce$ used during the \gls{ce} procedure, as defined in~\cite{Swindlehurst2022ce,Wang2020ce,Mo2023ce}. The \gls{ce} procedure is presented in Section~\ref{sec:protocol}.

\subsection{Channel modeling and control signaling}
\label{sec:channel-model}
As shown in Fig.~\ref{fig:scenario}, we make use of four different channels: a \gls{dc} used by the \glspl{ue} to transmit data to the \gls{ap} and three different \glspl{cc} -- \gls{ue}-\gls{cc}, \gls{ris}-\gls{cc}, and \gls{es}-\gls{cc} -- used to exchange the control information of the \glspl{ue}, the \gls{ris}, and the \gls{es}, respectively. In particular, the \gls{ue}-\gls{cc} represents the physical \gls{dl} and \gls{ul} \gls{cc} between the \gls{ap} and the \glspl{ue} (\emph{e.g.}, PDCCH/PUCCH in 5G~\cite{3gpp:nr}), which is used for informing about scheduling, reporting channel quality, transmitting \gls{ack} signals, etc. The \gls{ris}-\gls{cc} enables the \gls{ap} to control the \gls{ris} behavior. The \gls{es}-\gls{cc} coordinates the exchange of information over the backhaul link connecting the \gls{ap} and the \gls{es}. In general, we use the \emph{far-field assumption} when modeling the wireless channels, as in~\cite{DiLorenzo2022, saggese2023impact, croisfelt2023access}. Moreover, we assume a \textit{block-fading model}, \ie, the wireless channels are static and frequency-flat within each slot. The choice of the slot duration $\tau$ can be thus related to the coherence time of the channel.

We assume that control signaling occurs through the transmission of \textit{control packets} over the \gls{ue}-\gls{cc} and \gls{ris}-\gls{cc}. A control packet has a duration of $T$ seconds, defined as the system's minimum \gls{tti}.


\subsubsection{Data Channels (DCs)}
Due to \gls{fdm}, the $k$-th \gls{ue} transmits data to the \gls{ap} through its own \gls{dc}, which is a \gls{ul} channel operating with bandwidth $B_k$. 
The \gls{dc} channel is denoted as $\mb{h}_{k}(t, \bm{\phi}(t)) \in \mathbb{C}^{M}$ being a function of the slot $t$ and the corresponding \gls{ris} configuration loaded. It can be written as
\begin{equation} \label{eq:channel}
    \mb{h}_{k}(t, \mb{\phi}(t)) = \mb{d}_k(t) + \mb{G}_k(t) \bm{\phi}(t),
\end{equation}
where $\mb{d}_k(t)\in\mathbb{C}^{M}$ and $\mb{G}_k(t)\in\mathbb{C}^{M \times N}$ represent the direct \gls{ue}-\gls{ap} and the equivalent reflected \gls{ue}-\gls{ris}-\gls{ap} channel coefficients, respectively.

\subsubsection{ES Control Channel (ES-CC)}
The \gls{es}-\gls{cc} is considered to be an \gls{obcc}, \ie, the resources used for this channel are orthogonal w.r.t. the ones used for the \gls{dc}, resulting in an instantaneous and error-free channel. This is because a system designer can easily make this \gls{cc} as reliable as possible due to the high available pool of resources. Moreover, the \gls{es} and \gls{ap} are most likely co-located and cable-connected in our scenario.

\subsubsection{UE Control Channel (UE-CC)}\label{sec:channel-model:ue-cc}
The \gls{ue}-\gls{cc} of the $k$-th \gls{ue} works over both \gls{dl} and \gls{ul} directions and is an \gls{ibcc}, \ie, its operational frequency and bandwidth overlap with the resources used for the \gls{dc} of the $k$-th \gls{ue}. This is considered to let the control signaling sent by the \gls{ap} be aided by the \gls{ris} using the \gls{ctrl} configuration. Moreover, we assume control signaling is performed by exchanging control packets transmitted to or from the \glspl{ue}, whose packets are sent using bandwidth $B_k$.

\subsubsection{RIS Control Channel (RIS-CC)} 
We assume that the \gls{ap} controls the \gls{ris} over a wireless \gls{dl} connection. We analyze the \gls{ris}-\gls{cc} in the \gls{ibcc} condition that considers that the resources of this channel 
overlap with the ones of the \glspl{dc}. Thus, control signaling for \gls{ris} control cannot co-occur with transmissions for/from \glspl{ue}. 
It is assumed that control signaling on this \gls{cc} is performed through control packets transmitted by the \gls{ap} over bandwidth $B$. 

%
\subsection{Dynamic queuing modeling}
\label{sec:queue-model}
Following~\cite{DiLorenzo2022}, we consider a dynamic queuing system, where each \gls{ue} has a \emph{local communication queue} to buffer data to be transmitted/offloaded, while the \gls{es} has a \emph{remote computation queue} to buffer data to be processed. We present the general update rules of these queues below. 

Let $\tilde{R}_{k}(t)$ denote the \emph{actual throughput} of the $k$-th \gls{ue} at slot $t$, $k\in\mathcal{K}$. We denote as $Q^{u}_k(t)$ the local communication queue of the $k$-th \gls{ue}, updated as follows
\begin{equation}
    Q_{k}^u(t+1) = \max[0, Q_{k}^u(t) - \tau\pay \tilde{R}_{k}(t)] + A_{k}(t),
    \label{eq:dynamic:local-queue-update}
\end{equation}
where $A_k(t)$ [bits] is the realization of the data arrival process at slot $t$.

Let $f_{\max}$ be the total \gls{cpu} frequency of the \gls{es}. Let $Q^e_{k}$ denote the remote computation queue for the $k$-th \gls{ue} at the \gls{es}. We assume that the \gls{es} allocates computational resource $f_k(t)$ to process the task offloaded by the $k$-th \gls{ue} at slot $t$, such that $\sum_{k\in\mc{K}} f_k(t) = f_{\max}$. Thus, the remote computation queue for the $k$-th \gls{ue} at the \gls{es} evolves as
\begin{equation}
    \begin{aligned}
    Q^e_{k}(t+1) = &\max[0, Q^e_{k}(t) -  \tau\pay f_{k}(t) J_k] \\
    &+ \min \left( Q^u_{k}(t), \tau\pay \tilde{R}_{k}(t)\right),
    \end{aligned}
    \label{eq:dynamic:remote-queue-update}
\end{equation}
where $J_k$ [bit/Hz/s] is an efficiency parameter that depends on the application offloaded by the $k$-th \gls{ue}. 

The \emph{total queue} for the $k$-th \gls{ue} can be described as
\begin{equation} 
    Q_{k}(t)=Q_{k}^{u}(t)+Q_{k}^{e}(t).
    \label{eq:dynamic:total-queue}
\end{equation}
The \emph{average offloading latency} perceived by the $k$-th \gls{ue} is
\begin{equation} \label{eq:latency}
    L_k(t) = Q_{k}(t) / \bar{A}_k,
\end{equation}
where $\bar{A}_k = \mathbb{E}_{\sim t}[A_k(t)] / \tau$ is the average arrival rate [kbps]. 

\subsection{Energy consumption modeling}
\label{sec:energy-model}
The total energy consumption in the system is modeled as a weighted sum of the energy consumed by the \gls{es}, \gls{ap}, \gls{ris}, and the \glspl{ue} at each slot, denoted as $E_e(t)$, $E_a(t)$, $E_r(t)$, and $E_{k}(t)$, respectively. Based on~\cite{DiLorenzo2022}, we have the following total energy consumption at slot $t$
\begin{equation} \label{eq:total-energy}
    E_{\sigma}^\mathrm{tot}(t) =  \sigma \sum_{k\in\mc{K}} E_{k}(t) + (1-\sigma) (E_e(t) + E_a(t) + E_r(t)),
\end{equation}
where $\sigma \in [0,1]$ is a weighting parameter striking a trade-off between \gls{ue} and network consumption; choosing $\sigma =1$ leads to a pure user-centric strategy; whereas, $\sigma =0$ induces a pure network-centric strategy. Below, we model each one of the energy terms, where, different from~\cite{DiLorenzo2022}, we introduce the energy consumption related to control operations.

The energy consumption at the \gls{es} depends on the computation tasks over offloaded data and control operations. Modeling the energy spent for computation using a \gls{cmos}-based \gls{cpu}~\cite{burd1996processor}, we obtain that the \gls{es} consumes
\begin{equation}
    E_e(t) = \tau\pay\gamma_{c}f_{\max}^3 + E_e^\ctl(t),
\end{equation}
where $\gamma_{c}$ is the effective switching capacitance of the \gls{es} processor, and $E_e^\ctl$ is the \gls{es} energy consumption related to control.

The energy spent by the $k$-th \gls{ue} depends on the data transmission and control operations and is given by
\begin{equation}
    E_{k}(t) = \tau\pay P_k(t) + E_k^\ctl(t),
\end{equation}
where $P_k(t)$ is the power spent by the $k$-th \gls{ue} for data transmission, and $E^\ctl_{k}(t)$ is the $k$-th \gls{ue} energy consumption related to control for $k\in\mathcal{K}$.

Since we consider \gls{ul} data transmissions, the \gls{ap} spends energy only for control operations, given by
\begin{equation}
    E_a(t) = E_a^\ctl(t).
\end{equation}

The \gls{ris} energy consumption depends on the power dissipated by its active elements and control operations, which can be modeled as
\begin{equation}
    E_r(t) =\tau\pay P_{r}(b) \sum_{n\in\mc{N}} \alpha_{n}(t) + E_r^\ctl(t),
\end{equation}
where $P_{r}(b)$ is the power dissipated by each of the $N$ $b$-bit resolution elements of the \gls{ris}, $\alpha_n(t)\in\{0,1\}$ is the element activation state (see eq.~\eqref{eq:phi}), and $E_r^\ctl(t)$ is the \gls{ris} energy consumption related to control.

In Section~\ref{sec:performance}, we will specify the energy consumed due to control operations for all the communication nodes.

\section{Offloading Protocol with Control Operations}\label{sec:protocol}
In this section, we present the protocol dictating how the computation offloading service occurs in a \gls{ris}-aided \gls{mec} system when considering its control aspects. Before the beginning of the offloading procedure, we assume that the \glspl{ue} connect to the network through an \emph{access phase}. During this phase, the \gls{ap} grants access to the \glspl{ue} by allocating the bandwidths $B_k$'s for the \glspl{dc} according to frequency multiplexing. Moreover, the \gls{ap} broadcasts essential protocol details, including the slot duration $\tau$ and other pertinent information, to all the \glspl{ue}. The proper design of this phase is beyond the scope of this paper (\eg, see~\cite{croisfelt2023access}).

\begin{figure}
    \centering
    \includegraphics[width=\columnwidth]{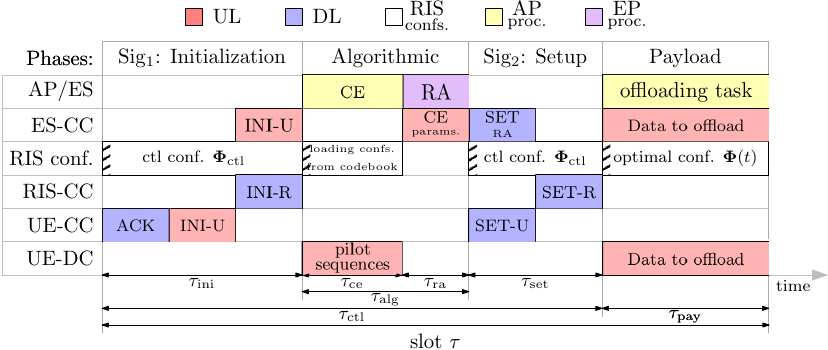}
    \caption{
    A timing diagram illustrating the offloading protocol execution at time slot $t$. The slot is divided into four phases: a) Signaling -- initialization, b) algorithmic, c) Signaling -- setup, and (d) payload. The first three phases occur within the control time $\tau_{\rm ctl}$ while the latter occurs over the payload time $\tau_{\rm pay}$. Control signaling is implemented through the introduction of four control packets: INI-U, with initial information sent by the \gls{ue} to the \gls{ap}; INI-R, with initial information sent by the \gls{ap} to the \gls{ris}; SET-U, with optimized \gls{ra} parameters sent by the \gls{ap} to the \gls{ue}; and SET-R, with an optimal configuration, $\bm{\phi}(t)$, sent by the \gls{ap} to the \gls{ris}. 
    }
    \label{fig:general}
\end{figure}
%


Fig.~\ref{fig:general} shows the timing diagram of the protocol within the $t$-th slot. In Section~\ref{sec:model}, we have defined that each slot of duration $\tau$ is divided into control $\tau_\ctl$ and payload $\tau\pay$ times (see eq.~\eqref{eq:slot}). Following the framework introduced in~\cite{saggese2023impact}, we further divide the control overhead into three different phases related to two different types of control operations: \emph{signaling} and \emph{algorithmic}. Thus, each slot consists of two signaling phases, \emph{initialization} and \emph{setup}, in which control information needed to initialize the network and set the optimized parameters of the offloading service are exchanged between nodes, respectively; an \emph{algorithmic} phase comprising the processing times of the \gls{ce} and \gls{ra} operations, which are needed to optimize the offloading service; and a \emph{payload} phase in which data transmission and offloading take place. The duration of these four phases is denoted as $\tau\ini$, $\tau\alg$, $\tau\set$, and $\tau\pay$, respectively. By definition, we have that $
    \tau_\ctl = \tau\ini + \tau\alg + \tau\set$,
and we can rewrite eq.~\eqref{eq:slot} as 
\begin{equation}
    \tau = \tau_\ctl + \tau\pay = \tau\ini + \tau\alg + \tau\set + \tau\pay.
    \label{eq:final-slot}
\end{equation}
Next, we detail each phase within the $t$-th slot. 

\subsection{Signaling: Initialization phase}
During this phase, the \gls{ris} loads the \gls{ctrl} configuration, $\bm{\phi}_\ctl$ (see Section \ref{sec:ris-model}). The \gls{ap} starts by carrying out an ACK procedure that informs the \glspl{ue} about the correct decoding of received data on the previous slot by broadcasting ACK/NACK packets, one packet per \gls{ue}. The \glspl{ue} prepare the data to offload accordingly, \ie, if the $k$-th \gls{ue} received a NACK, it re-inserts the data unsuccessfully transmitted as the oldest data in its queue $Q_k^u(t)$. Afterward, the \glspl{ue} inform the \gls{ap} on the state of their communication queues $Q_k^u(t)$, where each \gls{ue} sends a control packet called INI-U through the \gls{ue}-\gls{cc}. This information is then propagated to the \gls{es} through the \gls{es}-\gls{cc} because the \gls{es} needs it as input to perform the \gls{ra}. Meanwhile, the \gls{ap} informs the \gls{ris} to be prepared to start the \gls{ce} operation by sending a control packet called INI-R over the \gls{ris}-\gls{cc}.

\subsection{Algorithmic phase} 
The algorithm phase consists in the \gls{ce} and \gls{ra} procedures, having duration $\tau\ce$ and $\tau\ra$, respectively, with $\tau\alg = \tau\ce + \tau\ra$.

We start by describing the \gls{ce} procedure. The \gls{ce} operation takes place through the transmission of pilot sequences by the \glspl{ue} to the \gls{ap}. Then, the \gls{ap} uses the received signals to estimate the \gls{csi} comprised of direct channels $\mb{d}_k(t)$ and the equivalent reflected channels $\mb{G}_k(t)$, $\forall k\in\mc{K}$ by using one of the possible \gls{ce} methods for \gls{ris}-aided systems, \eg,~\cite{Swindlehurst2022ce, Wang2020ce, Mo2023ce}. In general, these methods estimate the channels using a two-step process: i) the \glspl{ue} transmit pilot sequences while the \gls{ris} is turned off, \ie, $\boldsymbol{\phi} = \mb{0}_N$; then, the \gls{ap} estimates the direct channel; ii) the \glspl{ue} transmit replicas of the pilot sequences while the \gls{ris} sweeps through a \emph{\gls{ce} codebook} $\mc{C}\ce$ containing \gls{ris} configurations $\boldsymbol{\phi}\in\mc{C}\ce$ -- a configuration is loaded per each pilot replica\footnote{Note that $\mc{C}\ce$ need to be shared among the \gls{ap} and the \gls{risc}, while the \glspl{ue} know its cardinality $|\mc{C}\ce| = C\ce$ to transmit an adequate number of pilot sequences; the information about $\mc{C}\ce$ and $C\ce$ can be exchanged during the access phase.}; then, the \gls{ap} can subtract the direct channel from the received signals and apply any classical estimation technique to estimate the reflected channels, taking care of removing the influences of the used configurations in $\mc{C}\ce$~\cite{Wang2020ce, Mo2023ce}. 
%
The overall estimated \gls{csi} for the $k$-th \gls{ue} is denoted as $\hat{\mb{h}}(t, \bm{\phi}) = \hat{\mb{d}}_k(t) + \hat{\mb{G}}_k(t) \bm{\phi}$, for any $\bm{\phi}$. The \gls{ap} shares the \gls{csi} knowledge obtained with the \gls{es}. 

We now describe the \gls{ra} procedure. The \gls{es} can perform a \gls{ra} procedure based on the estimated \gls{csi}, on the knowledge of the \glspl{ue} local communication queues given by the transmission of INI-U packets, and on its remote computation queues. We adopt the \gls{ra} procedure introduced in~\cite{DiLorenzo2022}. This procedure minimizes the total energy consumption while constraining the perceived average offloading latency under a threshold, $\bar{L}$. This is done by leveraging Lyapunov stochastic optimization in a greedy manner that optimizes the following parameters: 1) the \gls{ap} beamforming vector $\mb{w}(t)$, 2) the \glspl{ue} power coefficients denoted as $P_k(t)$, 3) the \gls{es} offloading computation resources $f_k(t)$, $\forall k\in\mc{K}$, and 4) the \gls{ris} configuration $\bm{\phi}(t)$ (see~\cite{DiLorenzo2022} for further details on the optimization). 

\subsection{Signaling: Setup phase}
During this phase, the \gls{ris} again loads the \gls{ctrl} configuration, $\bm{\phi}_\ctl$ (see Section~\ref{sec:ris-model}). This phase aims to send the parameters optimized by the \gls{ra}, carried out during the algorithmic phase by the \gls{es}, to the \glspl{ue} and the \gls{ris}. First, the \gls{es} sends the \gls{ra} parameters towards the \gls{ap} through the \gls{es}-\gls{cc}. Then, the \gls{ap} sends control packets called SET-U to each \gls{ue} containing its power coefficient $P_k(t)$ and its nominal throughput $R_k(t)$ over the \gls{ue}-\gls{cc}. 
Similarly, the \gls{ap} sends a control packet called SET-R to the \gls{ris} containing the configuration $\bm{\phi}(t)$ over the \gls{ris}-\gls{cc}.

\subsection{Payload phase} \label{sec:payload}
In this phase, every \gls{ue} offloads its data by transmitting with its power coefficient $P_k(t)$ and its nominal throughput $R_k(t)$ while the \gls{ris} loads the configuration $\bm{\phi}(t)$ and the \gls{ap} employs the beamforming vector $\mb{w}(t)$. Accordingly, it is expected that the $k$-th \gls{ue} transmits with power $P_k(t)$ and \emph{nominal throughput} $R_k(t)$ given by
\begin{equation} \label{eq:throughput}
    R_k(t) = B_k \log\left(1 + \frac{P_k(t)}{N_0 B_k} |\mb{w}(t)\H \hat{\mb{h}}_k(t, \bm{\phi}(t))|^2 \right),
\end{equation}
where $N_0$ [W/Hz] is the noise spectral density. 
In the ideal case of perfect \gls{csi} with error-free \glspl{cc}, the \gls{ue} can reliably transmit at its channel capacity, given by\footnote{We are assuming to work on the Shannon limit for simplicity of presentation. A real system should account for the loss provided by the \gls{mcs} used.}
\begin{equation} \label{eq:capacity}
    C_k(t) = B_k \log_2 \left(1 + \frac{P_k(t)}{N_0 B_k} |\mb{w}(t)\H \mb{h}_k(t, \boldsymbol{\phi}(t))|^2 \right).
\end{equation}
In this case, the actual throughput in~\eqref{eq:dynamic:local-queue-update} and the nominal throughput in~\eqref{eq:throughput} are the same as $\tilde{R}_k(t) = R_k(t) = C_k(t)$. 
However, we are interested in studying cases in which the control operations are not error-free. Specifically, errors in the \gls{ce} or the control signaling may lead to a nominal throughput higher than the channel capacity, resulting in the loss of the payload data. Hence, the actual throughput used to empty the \gls{ue} queues in~\eqref{eq:dynamic:local-queue-update} and~\eqref{eq:dynamic:remote-queue-update} evaluates to
 \begin{equation}
     \tilde{R}_k(t) = 
     \begin{cases}
         0,  &\mathrm{if}\,\, R_k(t) > C_k(t), \\
         R_k(t), &\mathrm{if}\,\, R_k(t) \le C_k(t).
     \end{cases}
 \end{equation}
We analyze the impact of the control operations in the next section.

\section{Impact of Control Operations}
\label{sec:performance}
In this section, we analyze the impact of control operations on the perceived offloading delay and the total energy consumption of the system. We first evaluate the overhead and the energy consumption introduced by the control operations; then, we describe how losing the control packets affects the decisions of the nodes.


%
\subsection{Overhead evaluation}
The control overhead $\tau_{\rm ctl}$ accounts for the time reserved for all control operations, which depends on the number of control packets sent during the initialization and setup phases and the duration of the \gls{ce} and \gls{ra} procedures during the algorithmic phase. Next, we characterize the control overhead of eq.~\eqref{eq:final-slot}.

\paragraph{Signaling: Initialization phase} 

The control overhead due to the initialization phase consists of i) the simultaneous transmission of $K$ one-bit messages for the \gls{ack} procedure, ii) the simultaneous transmission of $K$ INI-U control packets, and iii) the transmission of a single INI-R control packet. Moreover, we consider a guard period of $\tau_s$ seconds to load the \gls{ctrl} configuration. 
Thus, the initialization overhead is
\begin{equation}
    \tau\ini = \tau_s + 3 T.
\end{equation}

\paragraph{Algorithmic phase}
Due to \gls{fdm}, the \gls{ce} of each \gls{ue} is performed simultaneously.
In each bandwidth $B_k$, the $k$-th \gls{ue} transmits a pilot sequence composed of $N_p$ \glspl{tti} for the direct path estimation, and the same sequence is repeated for each \gls{ris} configuration in the \gls{ce} codebook $\mathcal{C}\ce$ for the reflected path estimation. The \gls{ce} overhead results in
\begin{equation}
    \tau\ce = (\tau_s + N_p T) (C\ce + 1),
\end{equation}
where $\tau_s$ is the switching time between each configuration and $C\ce$ is the number of \gls{ce} configurations.

The \gls{ra} overhead $\tau\ra$ is determined by the number of computational cycles required in \gls{ra} task $n\ra$ and by the \gls{cpu} cycles $f\ra$ [cycle/s] employed for computation, obtaining: 
\begin{equation}
  \tau\ra = \frac{n\ra}{f\ra},
\end{equation}
where $f\ra \le f_{\max}$. To evaluate $n\ra$, we only consider the cost of jointly optimizing the \gls{ris} configuration and \gls{ap} beamforming vector and ignore the power coefficient computation, the latter needed regardless of the presence of the \gls{ris}. Following
~\cite[Algorithm 1]{DiLorenzo2022}, the optimal $\bm{\phi}(t)$ and $\mb{w}(t)$ are obtained employing a greedy approach: keeping fixed an \gls{ap} beamforming vector $\mb{w}(t) \in \mc{C}_\mathrm{ap}$, we select the phase-shift that minimizes the objective function in \cite[eq. (26)]{DiLorenzo2022} for a group of $N_g < N$ elements, keeping fixed the phase-shifts of the other elements $N - N_g$. Considering that each element can be inactive or active with $2^b$ different phase-shift values, the number of possible choices is $2^b + 1$. Note that the $N_g$ elements grouped will load the same phase-shift.
This procedure is repeated $\forall \mb{w}(t) \in\mc{C}_\mathrm{ap}$ to find the pair $(\mb{w}(t), \bm{\phi}(t))$ that
greedily minimizes the objective function.
Accordingly, the number of computational cycles yields in
\begin{equation}
    n\ra = C_\mathrm{ap} (2^b +1) \frac{N}{N_{g}} \mu,
\end{equation}
where $\mu = 2 [ K (3 N / N_g + M +4) + 3]$ is the number of multiplications needed to evaluate the objective function for every pair $(\mb{w}(t), \bm{\phi}(t))$~\cite[Algorithm 1]{DiLorenzo2022}\footnote{We used the multiplication as the dominant operations for simplicity.}. Clearly, the performance boost provided by the \gls{ris} can be maximized by applying this procedure for each element, \ie, by setting $N_g = 1$, at the cost of increasing $\tau\ra$. 

\paragraph{Signaling -- setup phase}
The control overhead due to the setup phase consists of i) the simultaneous transmission of $K$ SET-U control packets and ii) the transmission of a single SET-R control packet. Again, we consider the guard period of $\tau_s$ seconds to load the \gls{ctrl} configuration. Thus, the setup overhead is:
\begin{equation}
    \tau\set = 2 \tau_s + 2T.
\end{equation}

\subsection{Energy consumption due to control operations}
We now characterize $E_e^\ctl(t)$, $E_k^\ctl(t)$, $E_a^\ctl(t)$ and $E_r^\ctl(t)$ introduced in Section~\ref{sec:energy-model}. For the \gls{es}, the energy spent during the control portion of the slot is the energy used to perform the \gls{ra}, modeled in the same way as the offloading task, obtaining
\begin{equation}
    E_e^\ctl(t) = \gamma_{c} f\ra^3 \tau\ra.
\end{equation}
During control signaling, each $k$-th \gls{ue} sends only a single INI-U packet; for \gls{ce}, it sends $C\ce + 1$ pilot sequences of length $N_p$ \glspl{tti}. Hence, the energy consumption during control is
\begin{equation}
    E_k^\ctl(t) = P_k^\ctl(t) T (1 + N_p (C\ce + 1)),
\end{equation}
where $P_k^\ctl(t)$ is the power used for control signaling. The \gls{ap} consumes power to send control packets. It sends the ACK and SET-U packets to all the $K$ \glspl{ue}, while the INI-R and SET-R packets to a single receiver; hence, it consumes
\begin{equation}
    E_a^\ctl(t) = 2 P_a^\ctl(t) T  (K + 1),
\end{equation}
where $P_a^\ctl(t)$ is the \gls{ap} power used for control signaling. Finally, we assume that all the \gls{ris} elements are active when loading $\bm{\phi}_\ctl$ and during the \gls{ce} of the reflected path, yielding 
\begin{equation}
    E_r^\ctl(t) = (\tau\ini +\tau\set + C\ce  N_p T) N P_{r}(b).
\end{equation}

\subsection{The effect of control errors}
The performance of the overall network is tied with the actual throughput $\tilde{R}_k(t)$, dependent on the nominal throughput $R_k(t)$ and the possible errors in the \gls{ce} and in the control signaling that can occur when decoding the control packets. 

\paragraph{Errors in \gls{ce}}
The following proposition characterizes the estimation error during \gls{ce}.
\begin{proposition}
    Assuming \gls{ls} estimation applied on pilot replicas of $N_p$ symbols length and uncorrelated Rayleigh channels, the \gls{csi} at the \gls{ap} after the \gls{ce} is
    \begin{equation}
    \begin{aligned}
        \hat{\mb{d}}_k(t) &= \mb{d}_k(t) + \mb{n}_{k}(t)\sim \mc{CN}(\mb{0}_M, \lambda_k \mb{I}_M), \\
        \hat{\mb{G}}_k(t) &= \mb{G}_k + \mb{N}_k(t)\sim \mc{CN}(\mb{0}_{MN}, \gamma_k \mb{I}_{MN}),
    \end{aligned}    
    \end{equation}
    having variances
    \begin{equation}
        \begin{aligned}
            \lambda_k = \frac{N_0 B_k}{N_p P_k^\ctl(t)}, \quad \gamma_k = \frac{2 N_0 B_k}{N_p P_k^\ctl(t)} \frac{N}{C\ce^2}.
        \end{aligned}
    \end{equation}
\end{proposition}
\begin{proof}
    The expression of the $\hat{\mb{d}}_k(t)$ is derived by correlating the received signal with the transmitted pilot when $\bm{\phi}(t) = \bm{0}_N$; then $\hat{\mb{G}}_k(t)$ is obtained by correlating the received signal with the pilot, subtracting $\hat{\mb{d}}_k$, and post-multiply it by the configurations in $\mc{C}\ce$~\cite{Wang2020ce}.
\end{proof}

\paragraph{Errors in control signaling}
If correct control occurs, \ie, no control packet is lost or erroneously decoded, the transmission power $P_k(t)$ and the nominal throughput $R_k(t)$ in~\eqref{eq:throughput} are correctly set by the \gls{ue}. In the following, we detail the effect on the decision made by the nodes caused by losing each of the control packets\footnote{The ACK packet is not analyzed because it is needed regardless of the presence of a \gls{ris} in the system. Hence, the ACK packet is always reliable, an assumption justified by its limited informative content.} -- INI-U, INI-R, SET-U, SET-R.

If the $k$-th \gls{ue}'s INI-U packet is lost, the information on the \gls{ue} local queue is not provided to the \gls{es}; then, the latter can infer from the received data that the \gls{ue} communication queue was emptied by $\tau\pay \tilde{R}_k(t)$ bits. Since the \gls{es} is still unaware of the amount of new data $A_k(t)$ arrived at the \gls{ue}, its queue is updated following 
\begin{equation}
    \hat{Q}^u_k(t+1) = \max[0, Q_k^u(t) - \tau\pay \tilde{R}_k(t)].
\end{equation}
The \gls{ra} process will use this input to optimize the power coefficient and the nominal throughput, impacting the offloading latency.

If the INI-R packet is not received, the \gls{ris} is unaware of the need for switching configurations for the \gls{ce} process. Then, the whole \gls{ce} procedure occurs while the \gls{ris} loads the control configuration $\bm{\phi}_\mathrm{ctrl}$. For the sake of simplicity, we assume that no \gls{csi} can be obtained in this case, and the nominal throughput of all \glspl{ue} result in $R_k(t)  = 0$, $\forall k\in\mc{K}$.

If the $k$-th \gls{ue}'s SET-U packet is lost, then the $k$-th \gls{ue} is not informed by the power coefficient $P_k(t)$ and the nominal throughput $R_k(t)$ in~\eqref{eq:throughput} for the payload transmission. Hence, the best the \gls{ue} can do is to transmit with the power and the throughput of the previous slot, hoping the data stream can be decoded at the \gls{ap}. We have
\begin{equation}
    P_k(t) = P_k(t-1) \text{ and } R_k(t) = R_k(t-1).
\end{equation}

If the SET-R packet is not received, the \gls{ris} is not informed of $\mb{\phi}(t)$ and can only resort to loading the optimal configuration of the previous slot. This affects the channel capacity of the system, being now, $\forall k\in\mc{K}$,
\begin{equation}
    C_k(t) = B_k \log_2 \Big(1 + \frac{P_k(t)}{N_0 B_k} |\mb{w}(t)\H \mb{h}_k(t, \boldsymbol{\phi}(t-1))|^2 \Big).
\end{equation}

Note that the reliability of the control packets depends on their informative content, the time reserved for the transmission, and the bandwidth of the \glspl{cc}. In this paper, the impact of reliability will be analyzed numerically.

\section{Numerical Results}

\begin{table}[bt]
    \centering
    \caption{Simulation parameters.}
    \footnotesize
    \begin{tabular}{lcc}
    \toprule    
    Scenario radius & $r$ & $100$~m  \\
    \gls{ap} position & $\mb{x}_a$ & $50 \sqrt{2} [1, 1, 0]\T$~m \\    
    \gls{ris} elements & $N$ & 64  \\
    \gls{ap} antennas & $M$ & 8 \\
    No. of \glspl{ue} & $K$ & 4 \\    
    No. of slots    & $N_t$ & 100 \\
    Overall bandwidth & $B$ & $500$~MHz\\
    \gls{ue}/\gls{ap} \gls{ctrl} power & $P_k^\ctl(t)$, $P_a^\ctl(t)$ & 20, 24  dBm \\
    \gls{ue}/\gls{ap} noise density & $N_0$ & $-170$~dBm/Hz\\
    Gain at reference distance & $\sigma_0^2$ & -38 dB m$^2$ \\
    Max. and \gls{ra} \gls{cpu} freq. & $f_{\max}$, $f\ra$ & 4.5. 0.5 GHz \\    
    Avg. arrival rate & $\bar{A}_k$ & [50, 100, 200] kbps \\
    Avg. latency constraint    & $\bar{L}$ & 300 ms \\
    Energy weighting parameter & $\sigma$ & 0.5 \\
    Codebooks cardinality & $C_\mathrm{ap}$, $C\ce$  & 25, $N$ \\
    Slot, \gls{tti}, \gls{ris} guard time & $\tau$, $T$, $\tau_s$ & $60:300$, 1/14, 0 ms \\    
    Pilot sequence length & $N_p$ & 1 \\
    Group of \gls{ris} elements & $N_g$ & $2$ \\ 
    \bottomrule
    \end{tabular}
    \label{tab:params}
    \vspace{-0.3cm}
\end{table}

In this section, we numerically evaluate the impact of control operations in the \gls{qos} of the computation offloading service in \gls{ris}-aided \gls{mec} systems. The presented results are obtained simulating a scenario similar to Fig.~\ref{fig:scenario}: the \gls{ris} is positioned at the center of a semi-circle of radius $r$, while the \gls{ap}'s position is $\mb{x}_a\in\mathbb{R}^{3}$; for every setup under test, the $K$ \glspl{ue} have random positions $\mb{x}_k\in\mathbb{R}^{3}$ within the semi-circle, and $N_t$ different slots are simulated. The presented results are averaged for the different setups if not specified differently. The channels are generated following Rayleigh fading, \ie, $\mb{d}_k(t) \sim \mc{CN}(\mb{0}_M, \sigma_0^2 \lVert \mb{x}_a - \mb{x}_k\rVert^{-3} \mb{I}_M)$, and $[\mb{G}_k(t)]_m = (\mb{r}_k(t) \circ  \mb{g}_m(t))\T$, $\forall m\in\mc{M}$, with $\mb{g}_m(t) \sim \mc{CN}(\mb{0}_{N}, \sigma_0^2\lVert \mb{x}_a\rVert^{-2} \mb{I}_{N})$ being the \gls{ris}-\gls{ap} channel, and $\mb{r}_k(t) \sim \mc{CN}(\mb{0}_N, \sigma_0^2\lVert \mb{x}_k\rVert^{-2} \mb{I}_N)$ the \gls{ue}-\gls{ris} one. The other relevant parameters are listed in Table~\ref{tab:params}~\cite{DiLorenzo2022, saggese2023impact}\footnote{See \url{https://github.com/victorcroisfelt/mec-with-ris-control} for the code.}. With the parameters under test, the overall control overhead results in $\tau_\ctl = 50.2$ ms, composed by $\tau\ini + \tau\set = 3.6$ ms, $\tau\ce = 46.4$ ms, and $\tau\ra \approx 0.17$ ms with $\tau\alg=\tau\ce+\tau\ra$. This demonstrates that the \gls{ce} impact on the overhead is dominant, constraining the whole system to work with time granularity that can hardly be met in practical scenarios. Hence, further research targeting the reduction of the \gls{ce} burden is needed to integrate the \gls{ris} in real systems.

\begin{figure}[hbt]
    \centering
    \input{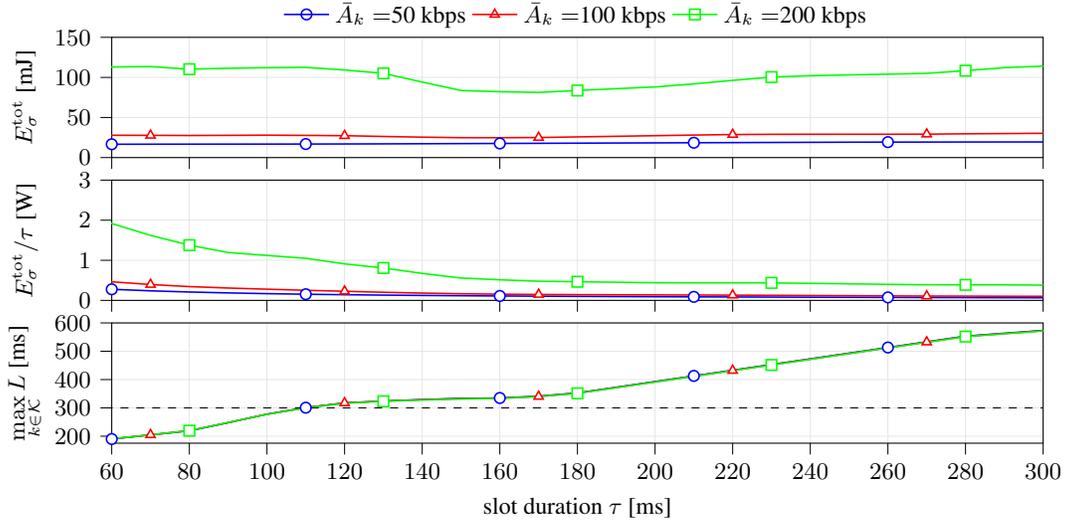}
    \caption{System performance vs. slot duration (error-free).}
    \label{fig:time-slot}
\end{figure}

Fig.~\ref{fig:time-slot} shows the average energy consumption, the average power spent per slot, and the maximum offloading latency as a function of the slot duration $\tau$, considering an error-free environment, \ie, with perfect \gls{csi} knowledge, and perfect control signaling. The parameters simulate a scenario where the coherence time is 300 ms, and the slot duration can be set to strike an energy-latency trade-off. Note that the control overhead is fixed, resulting in an increased payload time $\tau\pay$ when $\tau$ increases, following~\eqref{eq:final-slot}. 
While the selection of the slot duration slightly influences the average energy consumption, the power spent per slot reduces with increasing $\tau$ due to the possibility of decreasing the nominal throughput and the \gls{es} \gls{cpu} frequency while transmitting and processing the same amount of bits. On the other hand, increasing $\tau$ directly increases the offloading latency experienced, risking violating the latency constraint. Moreover, an eventual increase of $\bar{A}_k$ is absorbed by higher energy consumption without impacting the latency performance. Hence, the optimal working point is given by the highest slot time able to support the latency constraint, which is around $\tau = 100$ ms for our scenario.

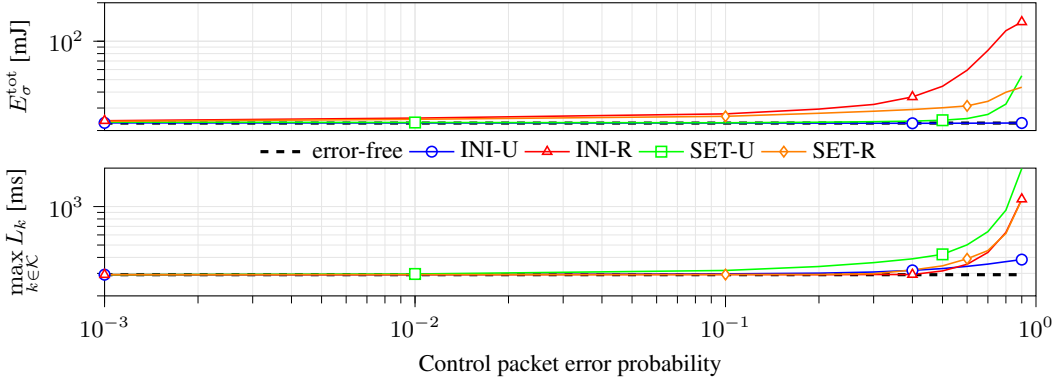
\begin{figure}[htb]
    \centering
%
%

\def\shift{-0.7cm}
\def\sep{0.5cm}
\def\vside{1.7cm}

\begin{tikzpicture}
\begin{groupplot}[
group style={group name=error proba, group size=1 by 2,  horizontal sep=\sep, vertical sep=\sep, x descriptions at=edge bottom, ylabels at=edge left}, 
height=\vside,
width=0.75\columnwidth,
scale only axis,
xmode=log,
xmin=0.001,
xmax=1,
xminorticks=true,
xmajorgrids,
xminorgrids,
ymajorgrids,
ymode=log,
yminorticks=true,
yminorgrids,
xlabel={Control packet error probability},
ytick={1e1,1e2,1e3}
]
\nextgroupplot[
legend style={  
  at={(0.5, 0)},
  anchor=north,  
  draw=none,
  fill opacity=0,
},
legend columns=5,
ymin=20,
ymax=200,
ylabel={$E_\sigma^\mathrm{tot}$ [mJ]},
xminorticks=false,
xmajorticks=false
]
\addplot [very thick, black, dashed]
  table[row sep=crcr]{%
0.001	22.873544170934\\
0.9	22.873544170934\\
};
\addlegendentry{error-free}

\addplot [semithick, blue,  mark=*, mark options={fill=white, solid}, mark repeat=5, mark phase=0]
  table[row sep=crcr]{%
0.001	22.9548583712455\\
0.01	22.9202327221325\\
0.1	22.8915645482797\\
0.2	22.859790747286\\
0.3	22.8099914459255\\
0.4	22.7872040101804\\
0.5	22.7935133211807\\
0.6	22.8220833873508\\
0.7	22.8476343805325\\
0.8	22.8803840179613\\
0.9	22.8969129337076\\
};
\addlegendentry{INI-U}

\addplot [semithick, red,  mark=triangle*, mark options={fill=white, solid}, mark repeat=5, mark phase=1]
  table[row sep=crcr]{%
0.001	23.8655030743458\\
0.01	25.0234679306422\\
0.1	27.0005878521576\\
0.2	29.4382357183279\\
0.3	32.0013889437218\\
0.4	36.6459653956521\\
0.5	44.3972184180906\\
0.6	59.2906441002966\\
0.7	84.7518534962227\\
0.8	120.900339925995\\
0.9	141.495942705976\\
};
\addlegendentry{INI-R}

\addplot [semithick, green,  mark=square*, mark options={fill=white, solid}, mark repeat=5, mark phase=2]
  table[row sep=crcr]{%
0.001	23.2785372129381\\
0.01	23.1100119747596\\
0.1	23.0037463853459\\
0.2	23.2214821543221\\
0.3	23.4641330734204\\
0.4	23.7004527522772\\
0.5	24.076810108578\\
0.6	24.8160493175984\\
0.7	26.7340066704191\\
0.8	32.1949370113494\\
0.9	53.6961610986849\\
};
\addlegendentry{SET-U}

\addplot [semithick, orange,  mark=diamond*, mark options={fill=white, solid}, mark repeat=5, mark phase=3]
  table[row sep=crcr]{%
0.001	23.6001640048013\\
0.01	24.4784901545798\\
0.1	25.82926060073\\
0.2	27.2952765363434\\
0.3	28.3187929368941\\
0.4	29.200270249749\\
0.5	30.1480924694199\\
0.6	31.1714442999339\\
0.7	33.8500219926425\\
0.8	39.8816047524371\\
0.9	43.63453355499\\
};
\addlegendentry{SET-R}

\nextgroupplot[
ymin=200,
ymax=2000,
ylabel={$\displaystyle \max_{k\in\mc{K}} L_k$ [ms]},
yminorgrids,
yminorticks=true,
ymode=log,
ylabel shift={-1.5mm}
]
\addplot [very thick, black, dashed]
  table[row sep=crcr]{%
0.001	293\\
0.9 	293\\
};

\addplot [semithick, blue,  mark=*, mark options={fill=white, solid}, mark repeat=5, mark phase=0]
  table[row sep=crcr]{%
0.001	293\\
0.01	293.111873239438\\
0.1	297.423393651024\\
0.2	301.198144347043\\
0.3	307.465178244507\\
0.4	316.711280635493\\
0.5	327.223371758873\\
0.6	340.51745904676\\
0.7	354.626989701408\\
0.8	370.739625374647\\
0.9	384.667703170704\\
};

\addplot [semithick, red,  mark=triangle*, mark options={fill=white, solid}, mark repeat=5, mark phase=1]
  table[row sep=crcr]{%
0.001	293\\
0.01	293.4812509007\\
0.1	293.872632887888\\
0.2	293.012201641827\\
0.3	293.961104094649\\
0.4	293.750447310105\\
0.5	313.678050681691\\
0.6	353.285619042255\\
0.7	437.498167936647\\
0.8	627.877974828991\\
0.9	1143.12285820146\\
};

\addplot [semithick, green,  mark=square*, mark options={fill=white, solid}, mark repeat=5, mark phase=2]
  table[row sep=crcr]{%
0.001	293\\
0.01	296.788983098593\\
0.1	316.195837566198\\
0.2	339.549470350424\\
0.3	364.15810330141\\
0.4	389.901530523945\\
0.5	422.702273261973\\
0.6	501.584611145916\\
0.7	636.242666224344\\
0.8	936.586216901409\\
0.9	1984.11849329578\\
};

\addplot [semithick, orange,  mark=diamond*, mark options={fill=white, solid}, mark repeat=5, mark phase=3]
  table[row sep=crcr]{%
0.001	293\\
0.01	293.194315560565\\
0.1	293.755188905181\\
0.2	294.178834227742\\
0.3	298.202549813574\\
0.4	318.845942938067\\
0.5	343.55417938932\\
0.6	390.241782634318\\
0.7	452.417597815363\\
0.8	616.075376050139\\
0.9	1131.82514804774\\
};

\end{groupplot}

\end{tikzpicture}%
    \caption{System performance vs. control packet error probability ($\tau = 100$ ms, $\bar{A}_k = 50$ kbps).}
    \label{fig:error-proba}
\end{figure}

Fig.~\ref{fig:error-proba} shows the system performance as a function of the control packet error probability, assuming that each curve simulates that only one kind of control packet can be lost. Losing INI-U packets influences the performance the least, showing that the system can generally recover the arrival rate information in subsequent slots. Losing INI-R packets strongly affects energy consumption while moderately influencing latency because the system tends to spend more energy to offload more data after an unsuccessful \gls{ce}. Losing SET-U packets strongly affects the latency performance due to the use of the inaccurate value of $P_k(t)$ and $R_k(t)$, while an inaccurate $\bm{\phi}(t)$ is less harmful (see SET-R).


\begin{figure}[htb]
    \centering
%
%
\definecolor{mycolor1}{rgb}{0.00000,0.44700,0.74100}%
\definecolor{mycolor2}{rgb}{0.85000,0.32500,0.09800}%
\definecolor{mycolor3}{rgb}{0.92900,0.69400,0.12500}%
\definecolor{mycolor4}{rgb}{0.49400,0.18400,0.55600}%
\def\shift{-0.7cm}
\def\sep{0.52cm}
\def\vside{1.5cm}

\begin{tikzpicture}
\begin{groupplot}[
group style={group name=slot-time, group size=1 by 2,  horizontal sep=\sep, vertical sep=\sep, x descriptions at=edge bottom, ylabels at=edge left}, 
height=\vside,
width=0.8\columnwidth,
scale only axis,
xmin=64,
xmax=128,
xlabel={$C_\mathrm{ce}$},
xmajorgrids,
ymajorgrids,
xmajorgrids,
ymajorgrids,
legend style={at={(0.5, 0)}, anchor=north, draw=none, fill opacity=0, font=\scriptsize},
legend columns = 5,
]
\nextgroupplot[
ymin=0,
ymax=0.3,
ylabel={$E_\sigma^\mathrm{tot}$ [J]},
xmajorticks=false,
]
\addplot [semithick, blue,  mark=*, mark options={fill=white, solid}]
  table[row sep=crcr]{%
64	0.0335120028684774\\
80	0.0344164765948507\\
96	0.0421439878385395\\
112	0.104856541711509\\
128	0.0462389347196486\\
};
\addlegendentry{$\mu = $ 0.1}

\addplot [semithick, red,  mark=triangle*, mark options={fill=white, solid}]
  table[row sep=crcr]{%
64	0.0459822548032234\\
80	0.0461102255362508\\
96	0.0517047894525786\\
112	0.116696049617972\\
128	0.0796872529146348\\
};
\addlegendentry{$\mu = $ 0.3}

\addplot [semithick, green,  mark=square*, mark options={fill=white, solid}]
  table[row sep=crcr]{%
64	0.0699300651512684\\
80	0.0715574956453651\\
96	0.0780706535406401\\
112	0.141163607502573\\
128	0.149577716059976\\
};
\addlegendentry{$\mu = $ 0.5}

\addplot [semithick, orange,  mark=diamond*, mark options={fill=white, solid}]
  table[row sep=crcr]{%
64	0.163702359413277\\
80	0.165761540388707\\
96	0.17991190030339\\
112	0.24520340898755\\
128	0.288990906143447\\
};
\addlegendentry{$\mu = $ 0.7}

\addplot [semithick, violet,  mark=halfsquare*, mark options={fill=white, solid}]
  table[row sep=crcr]{%
64	0.216631046769824\\
80	0.216593358620242\\
96	0.216586694154289\\
112	0.216587422777956\\
128	0.216596457356067\\
};
\addlegendentry{$\mu = $ 0.9}

\nextgroupplot[
ymode=log,
ymin=0.2,
ymax=10,
yminorticks=true,
ylabel={$\displaystyle\max_{k\in\mathcal{K}} L_k$ [ms]},
yminorgrids,
ylabel shift = {-2mm},
ytick={1e-1, 1e0, 1e1, 1e2},
yticklabels={$10^{2}$, $10^3$, $10^4$}
]
\addplot [semithick, blue,  mark=*, mark options={fill=white, solid}, forget plot]
  table[row sep=crcr]{%
64	1.12609285730741\\
80	0.493328507013741\\
96	0.643122317222585\\
112	3.26458601689545\\
128	7.23157390417647\\
};
\addplot [semithick, red,  mark=triangle*, mark options={fill=white, solid}, forget plot]
  table[row sep=crcr]{%
64	1.12145018498343\\
80	0.473525915759457\\
96	0.46237954306067\\
112	1.5236186368344\\
128	6.6039258291233\\
};
\addplot [semithick, green,  mark=square*, mark options={fill=white, solid}, forget plot]
  table[row sep=crcr]{%
64	1.11784262454337\\
80	0.469323999004947\\
96	0.405521654413671\\
112	0.886315622082248\\
128	5.05523959325016\\
};
\addplot [semithick, orange,  mark=diamond*, mark options={fill=white, solid}, forget plot]
  table[row sep=crcr]{%
64	1.11738680510845\\
80	0.462093723107117\\
96	0.379704571168102\\
112	0.653883530749783\\
128	3.14250416640866\\
};
\addplot [semithick, violet,  mark=halfsquare*, mark options={fill=white, solid}, forget plot]
  table[row sep=crcr]{%
64	4.30482412007327\\
80	4.22353452705462\\
96	4.20858801865869\\
112	4.2303424940031\\
128	4.46613563195919\\
};
\addplot [color=black, dashed, forget plot]
  table[row sep=crcr]{%
64	0.3\\
128	0.3\\
};
\end{groupplot}
\end{tikzpicture}%
    \caption{System performance vs. the number of \gls{ce} configurations ($\tau=100$ ms, $\bar{A}_k = 100$ kbps).}
    \label{fig:ce}
\end{figure}
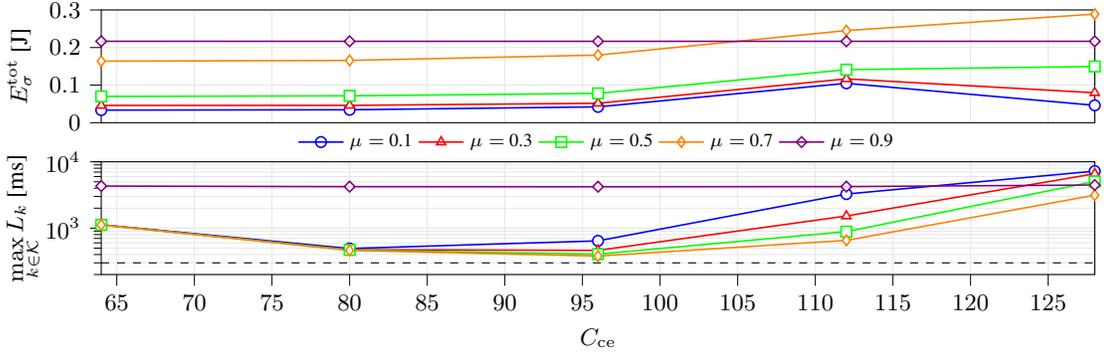

Finally, Fig.~\ref{fig:ce} shows the performance when accounting for the \gls{ce} errors as a function of the number of \gls{ce} configurations, $C\ce$. 
In this case, to reduce the probability that $R_k(t) > C_k(t)$, the nominal throughput is multiplied by a discount factor $\mu \le 1$. Under these conditions, the system is not able to provide the desired \gls{qos} regardless of the value of $\mu$ and $C\ce$, demonstrating the need to account for the control errors during the \gls{ra}, eventually optimizing the value of $\mu$\footnote{Remark that the \gls{ce} of the direct channels is influenced by $N_p$ only, which should also be set to a proper value in a real system.}. When a low value of $\mu$ is used, $\mu R_k(t) \le C_k(t)$ is probable, but the resulting throughput is insufficient to follow the arrival rate process. When $\mu$ is close to 1, it is likely that $\mu R_k(t) > C_k(t)$, resulting in a failed offload operation on slot $t$. Also, the results confirm the expected: 
a small codebook leads to worse \gls{ce} performance but higher payload time while increasing $C\ce$ leads to a lower payload time but better \gls{ce} performance.

\section{Conclusions}
In this paper, we have analyzed the control aspects related to \gls{mec} aided by \gls{ris}. Our investigation has highlighted the necessity of considering the various control operations when designing optimization algorithms for these services. In particular, the results demonstrate that the \gls{ce} procedure has the dominant impact on the \glspl{ue} \gls{qos}, generating the highest overhead and greatly influencing the \gls{ra} due to the \gls{csi} errors. Future works will capitalize on these findings, investigating possible robust \gls{ra} strategies along with the control operations.

\bibliographystyle{IEEEtranNoUrl}
\bibliography{IEEEabbr, main}

\end{document}